\documentclass[11pt]{amsart}
\usepackage[T1]{fontenc}
\usepackage{lmodern}
\usepackage{booktabs}
\usepackage{microtype}
\usepackage[hidelinks]{hyperref}
\hypersetup{pdftitle={A quaternary Legendre pair of length 64},
  pdfauthor={Harshit Verma}}

\newtheorem{theorem}{Theorem}
\DeclareMathOperator{\PAF}{PAF}
\newcommand{\ii}{\mathrm{i}}

\begin{document}

\title{A quaternary Legendre pair of length 64}
\author{Harshit Verma}
\address{Department of Computer Science, Yale University, New Haven, CT, USA}
\email{harshit.verma@yale.edu}

\begin{abstract}
A quaternary Legendre pair of length $\ell$ is a pair of sequences over
$\{\pm1,\pm\ii\}$ whose periodic autocorrelations sum to $-2$ at every
nonzero shift. We give a quaternary Legendre pair of length $64$. To the
best of our knowledge, this is the first reported construction at this
length.
\end{abstract}

\maketitle

\section{Introduction}

For a sequence $X=(X_0,\dots,X_{\ell-1})$ of complex numbers, with indices
taken modulo $\ell$, the periodic autocorrelation at shift $k$ is
\[
  \PAF(X,k)=\sum_{j=0}^{\ell-1}X_j\overline{X_{j+k}}.
\]
A \emph{quaternary Legendre pair} of length $\ell$ is a pair $(A,B)$ of
sequences in $\{\pm1,\pm\ii\}^{\ell}$ such that
\begin{equation}\label{eq:qlp}
  \PAF(A,k)+\PAF(B,k)=-2\qquad(1\le k\le\ell-1).
\end{equation}
When both sequences take values in $\{\pm1\}$, these are the Legendre pairs
of Fletcher, Gysin and Seberry~\cite{FGS2001}, which give Hadamard matrices
of order $2\ell+2$ and exist only for odd $\ell$. Kotsireas and
Winterhof~\cite{KW2024} introduced the quaternary version, which can also
exist for even $\ell$, and showed that a quaternary Legendre pair of length
$\ell$ gives a quaternary Hadamard matrix of order $2\ell+2$. They
conjectured that quaternary Legendre pairs exist for every even length. As
explained in~\cite[Section~2]{JP2025}, this conjecture would imply both
Turyn's conjecture that quaternary Hadamard matrices exist for every even
order and the Hadamard conjecture.

Jedwab and Pender~\cite{JP2025} constructed quaternary Legendre pairs of
length $(q-1)/2$ for every prime power $q\equiv1\pmod4$, and of length $2p$
for every odd prime $p$ such that $2p-1$ is a prime power. Together with the
earlier results of~\cite{KKW2025,KW2024}, this left fourteen even lengths up
to $100$ open~\cite{JP2025,KKW2025}:
\[
  42,\ 46,\ 52,\ 58,\ 64,\ 66,\ 70,\ 72,\ 76,\ 80,\ 88,\ 92,\ 94,\ 100.
\]
Lebedev~\cite{Lebedev2026} recently found pairs of lengths $42$, $46$, $52$,
$58$, $66$, $72$ and $80$. He also reported that searches at lengths $64$,
$70$ and $76$ within a family of reflection-symmetric (\emph{reversible})
pairs found none. In this note we settle length $64$.

\begin{theorem}\label{thm:main}
There is a quaternary Legendre pair of length $64$.
\end{theorem}

\section{The pair}

Let $A_j=\ii^{a_j}$ and $B_j=\ii^{b_j}$ for $0\le j\le63$, so that the digits
$0,1,2,3$ below stand for $1,\ii,-1,-\ii$. The exponents are
\[
\begin{array}{r@{\quad}l}
  a_0,\dots,a_{31}\colon    & \texttt{02022220 02022222 02200020 00202220}\\
  a_{32},\dots,a_{63}\colon & \texttt{20220002 02002000 00220022 00022022}\\[4pt]
  b_0,\dots,b_{31}\colon    & \texttt{02332221 10130032 10300312 13202110}\\
  b_{32},\dots,b_{63}\colon & \texttt{30112023 12130030 12300310 11222332}
\end{array}
\]

\begin{table}[ht]
\centering
\small
\setlength{\tabcolsep}{3pt}
\begin{tabular}{@{}l*{16}{r}@{}}
\toprule
$k$ & 1 & 2 & 3 & 4 & 5 & 6 & 7 & 8 & 9 & 10 & 11 & 12 & 13 & 14 & 15 & 16\\
\midrule
$\PAF(A,k)$ & $-4$ & $-4$ & $0$ & $4$ & $4$ & $0$ & $-8$ & $0$
  & $4$ & $0$ & $0$ & $-4$ & $4$ & $-4$ & $4$ & $4$\\
$\PAF(B,k)$ & $2$ & $2$ & $-2$ & $-6$ & $-6$ & $-2$ & $6$ & $-2$
  & $-6$ & $-2$ & $-2$ & $2$ & $-6$ & $2$ & $-6$ & $-6$\\
\midrule
$k$ & 17 & 18 & 19 & 20 & 21 & 22 & 23 & 24 & 25 & 26 & 27 & 28 & 29 & 30
  & 31 & 32\\
\midrule
$\PAF(A,k)$ & $0$ & $-4$ & $4$ & $-12$ & $0$ & $-4$ & $0$ & $0$
  & $-4$ & $4$ & $0$ & $0$ & $-4$ & $-8$ & $0$ & $-8$\\
$\PAF(B,k)$ & $-2$ & $2$ & $-6$ & $10$ & $-2$ & $2$ & $-2$ & $-2$
  & $2$ & $-6$ & $-2$ & $-2$ & $2$ & $6$ & $-2$ & $6$\\
\bottomrule
\end{tabular}
\caption{Periodic autocorrelations of $A$ and $B$ for $1\le k\le32$.}
\label{tab:paf}
\end{table}

\begin{proof}[Proof of Theorem~\ref{thm:main}]
Table~\ref{tab:paf} lists $\PAF(A,k)$ and $\PAF(B,k)$ for $1\le k\le32$.
All values are real, and in each column they sum to $-2$. For any sequence
$X$ of length $\ell$ we have $\PAF(X,\ell-k)=\overline{\PAF(X,k)}$, so
\eqref{eq:qlp} also holds for $33\le k\le63$.
\end{proof}

The entries of Table~\ref{tab:paf} are a direct computation in integer
arithmetic, since $A_j\overline{A_{j+k}}=\ii^{a_j-a_{j+k}}$. A short script
that checks~\eqref{eq:qlp} in this way is included with the arXiv version of
this note as the ancillary file \texttt{verify.py}.

\section{Remarks}

\subsection*{Normal form and Hadamard matrices}
The sequence $A$ has $32$ entries equal to $1$ and $32$ equal to $-1$, and
$B$ contains $1,\ii,-1,-\ii$ exactly $17$, $16$, $16$ and $15$ times, so
\[
  \sum_{j=0}^{63}A_j=0,\qquad\sum_{j=0}^{63}B_j=1+\ii.
\]
This is the normalization of~\cite[Lemma~2.1]{KW2024}, so
\cite[Theorem~2.2]{KW2024} applies to $(A,B)$ directly and gives a
quaternary Hadamard matrix of order $130$. A quaternary Hadamard matrix of
order $130$ was first constructed recently, by a different method,
in~\cite{KMTZ2026}.

\subsection*{Symmetry}
The sequence $B$ is symmetric: $B_{64-j}=B_j$ for $1\le j\le63$. The binary
sequence $A$, however, has no axis of symmetry: there is no $c$ with
$A_{c-j}=A_j$ for all $j$, indices taken modulo $64$. Hence no cyclic shift
of the pair is reversible in the sense of~\cite[Section~3]{Lebedev2026},
which is consistent with the negative searches reported there at
length $64$.

\bibliographystyle{amsplain}
\bibliography{references}

\end{document}